\documentclass[12pt]{article}
\usepackage{amsmath,amsfonts, amssymb,amsthm,graphicx,geometry,authblk,color,soul,comment,hyperref, algorithm, algpseudocode,microtype,caption} 
\newtheorem{theorem}{Theorem}
\newtheorem{remark}[theorem]{Remark}

 \title{On recovering intragranular strain fields from grain-averaged strains obtained by high-energy X-ray diffraction microscopy}
\author[1]{C.\,K.\ Cocke}
\author[1,2]{A.\ Akerson}
\author[1]{S.\,F.\ Gorske}
\author[1]{K.\,T.\ Faber}
\author[1]{K.\ Bhattacharya\thanks{Corresponding author: Kaushik Bhattacharya, bhatta@caltech.edu}}
\affil[1]{Division of Engineering and Applied Science, California Institute of Technology, Pasadena, CA 91125, USA}
\affil[2]{Reality Labs Research, Meta Platforms, Inc., Redmond, WA 98052, USA}
\date{}

\begin{document}

\maketitle

  \begin{abstract}
 We address an unusual problem in the theory of elasticity motivated by the problem of reconstructing the strain field from partial information obtained using X-ray diffraction.  Referred to as either high-energy X-ray diffraction microscopy~(HEDM) or three-dimensional X-ray diffraction microscopy~(3DXRD), these methods provide diffraction images that, once processed, commonly yield detailed grain structure of polycrystalline materials, as well as grain-averaged elastic strains.  However, it is desirable to have the entire (point-wise) strain field.  So we address the question of recovering the entire strain field from grain-averaged values in an elastic polycrystalline material.  The key idea is that grain-averaged strains must be the result of a solution to the equations of elasticity and the overall imposed loads.  In this light, the recovery problem becomes the following: find the boundary traction distribution that induces the measured grain-averaged strains under the equations of elasticity.  We show that there are either zero or infinite solutions to this problem, and more specifically, that there exist an infinite number of kernel fields, or non-trivial solutions to the equations of elasticity that have zero overall boundary loads and zero grain-averaged strains.  We define a best-approximate reconstruction to address this non-uniqueness.  We then show that, consistent with Saint-Venant's principle, in experimentally relevant cylindrical specimens, the uncertainty due to non-uniqueness in recovered strain fields decays exponentially with distance from the ends of the interrogated volume.  Thus, one can obtain useful information despite the non-uniqueness.  We apply these results to a numerical example and experimental observations on a transparent aluminum oxynitride~(AlON) sample.
  \end{abstract}

 \newpage

\section{Introduction}

Over the last 30 years, X-ray diffraction--based methods have become a powerful tool for the 3D characterization of crystalline materials.  Unlike destructive techniques such as serial sectioning, X-ray diffraction-based methods are nondestructive and allow for in situ characterization of materials during mechanical testing. Of particular interest here is high-energy X-ray diffraction (HEDM)~\cite{bernierHighEnergyXRayDiffraction2020}, which is generally performed at synchrotron light sources that offer high spectral brightness. 

Typically, two types of HEDM experiments are performed: far-field~(ff-HEDM) and near-field~(nf-HEDM), and the precise experimental and data processing details can be found elsewhere~\cite{suterForwardModelingMethod2006,bernierFarfieldHighenergyDiffraction2011,lienertHighenergyDiffractionMicroscopy2011,sharmaFastMethodologyDetermine2012,sharmaFastMethodologyDetermine2012a,bernierHighEnergyXRayDiffraction2020}.  In both methods, the sample is illuminated by a monochromatic X-ray beam and rotated. As the sample rotates, lattice planes within constituent crystallites diffract the incident beam as the Bragg condition is satisfied. Plane detectors placed downstream of the sample then measure the summed intensity of all diffraction events during a rotational increment. The difference between ff-HEDM and nf-HEDM is the detector distance, which controls how crystallite position manifests itself in diffraction measurements.  
Performing reconstruction from both measurement techniques and combining the data yields a complete three-dimensional grain orientation map and the grain-averaged elastic strain.  Although these measurements are insightful, it would be very useful to recover the actual (point-wise) elastic strain field.

This has led to the development of a point-focused HEDM~(pf-HEDM) (also known as scanning 3DXRD) technique that uses a 1D pencil beam to probe small volumes to obtain voxelized intragranular elastic strains~\cite{hayashiIntragranularThreedimensionalStress2019,hayashiScanningThreeDimensionalXRay2017,henningssonReconstructingIntragranularStrain2020,liResolvingIntragranularStress2023,henningssonIntragranularStrainEstimation2021}. 
Improved algorithms have also been developed to locally recover voxelized intragranular strains from near-field measurements~\cite{shenVoxelbasedStrainTensors2020,reischigThreedimensionalReconstructionIntragranular2020}. While these methods provide intragranular strains, they still average over the volume of the pencil beam, so the problem of obtaining the actual (point-wise) elastic strain field remains.

There have been attempts to address this as a post-processing step by incorporating additional physics.  Since the elastic moduli of the constituent materials are known, the voxel-averaged elastic stresses can be obtained from the corresponding strains.  These do not satisfy equilibrium, and so Zhou {\it et al.}~\cite{zhouImposingEquilibriumExperimental2022} sought to reconstruct the actual stress field by projecting the grain-averaged stresses to the nearest stress field satisfying equilibrium.  This work did not address strain compatibility, and therefore it is unclear whether the reconstructed stress field is physically meaningful.   A more complete reconstruction by Naragani {\it et al.}~\cite{naraganiInterpretationIntragranularStrain2021} in the context of plasticity used an iterative approach aiming to satisfy equilibrium, compatibility, and an elastic--plastic constitutive relation. However, this required a guess for the boundary condition and left the effects of nonuniqueness unexplored.  A difficulty in both works is that X-ray techniques can measure only the elastic strain, so the total strain is unknown.

\begin{figure}[t]
  \centering
  \includegraphics[width=6in]{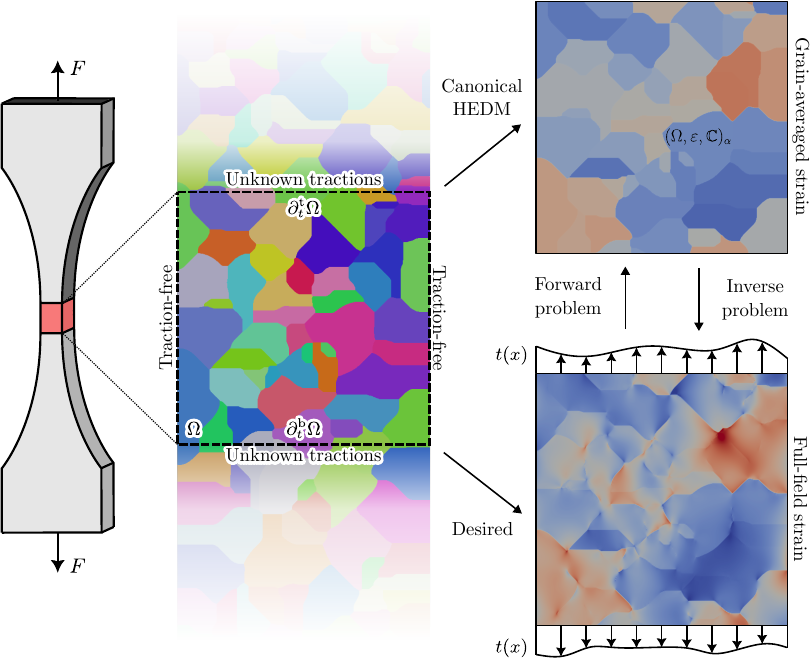}
  \caption{High-energy X-ray diffraction provides the grain structure (middle) and the grain-averaged strains (upper right) of a small cylindrical region of the specimen (left). We seek the full strain field, or equivalently the traction distribution on the top and bottom of the gauge section which induces a strain field satisfying the equations of elasticity (lower right).}
  \label{fig:problem_diagram}
\end{figure}

In this paper, we examine the mathematical status of the inverse problem of obtaining the actual (pointwise) strain field from the grain- or voxel-averaged values and the macroscopic applied load.  We limit ourselves to the simpler problem of elasticity, since one can measure the actual grain-wise elastic strains. Additionally, we study grain-averaged strains, although our methods may also be applied to voxel-averaged strains. The problem is shown schematically in Figure~\ref{fig:problem_diagram}. Note that the X-ray measurements are confined to a small region of the gauge section highlighted in the figure.

Specifically, we assume that we have a cylindrical domain with a square cross section that is loaded at the ends with an unknown traction distribution but with a known macroscopic force.  We also assume that we know the elastic modulus at each point and the strain averaged over a series of subdomains (grains).  We ask whether, given the grain-averaged strains and macroscopic force, we can find a strain field consistent with the equations of elasticity. To this end, the problem is cast as a problem of finding the boundary traction distribution that, through the equations of elasticity, induces the measured grain-averaged strains.

We have two main results.  First, we show that this problem either has no solution or infinite solutions (see Theorem~\ref{th:0inf}).  Trivially, one needs a consistency condition between the applied load and grain-averaged strains to show the existence of solutions.  Less trivially, we show the existence of an infinite number of kernel fields, or traction distributions with zero macroscopic force constructed such that their associated strain fields have zero grain-averaged strains. This follows from a counting argument.  The equations of elasticity give rise to a linear operator that maps (equilibrated) traction distributions to grain-averaged strains, and we show that this operator has an infinite-dimensional kernel.  We provide numerical examples of these kernel fields.  We also define a best-approximate reconstruction and apply it to experimental results on aluminum oxynitride~(AlON).

Our second result characterizes the non-uniqueness in the solution, i.e., the kernel fields. We observe from numerical results that the kernel tractions are highly oscillatory with the associated kernel strains decaying away from the ends of the cylinder.  This is a manifestation of Saint-Venant's principle~\cite{toupinSaintVenantsPrinciple1965}, though in a heterogeneous medium, and we seek to characterize these kernel strains in a central segment of the cylinder. We show that the largest possible kernel strain field (in the $L^2$-norm) decays exponentially as the distance between the central segment and the cylinder ends increases. It follows that all solutions of the strain reconstruction will agree far away from the ends, and that the strain reconstruction from grain-wise averages can be meaningful in the central segment even though the reconstruction is non-unique.

We formulate the problem and present the main mathematical result concerning the reconstruction of full-field strains given partial information and the equations of elasticity in Section~\ref{sec:problem_formulation}.  In Section~\ref{sec:examples}, we then introduce a numerical method to study this problem and illustrate the behavior of solutions before providing a series of computational examples. Finally, in Section~\ref{sec:saint_venants_principle} we address how the behavior of non-uniqueness follows Saint-Venant's principle before providing summary remarks in Section~\ref{sec:conclusions}.

\section{Full-field strains through mechanical admissibility} \label{sec:problem_formulation}

Let $\Omega \subset {\mathbb R}^d$ be the domain of interest (gauge section marked on the left of Figure~\ref{fig:problem_diagram}) consisting of $N_\mathrm{g}$~disjoint grains $\{\Omega_g\}_{g=1}^{N_\mathrm{g}} $, and let $\partial_t \Omega = \partial^\mathrm{t}_t \Omega \cup \partial^\mathrm{b}_t \Omega$ be two disjoint parts of the boundary to which traction is applied (see Figure~\ref{fig:problem_diagram}).  We assume that the rest of the boundary is traction-free and the medium is linear elastic with a known piecewise constant and positive-definite elastic modulus field~${\mathbb C}(x)$.   

The HEDM experiment gives us the total applied axial force~$F \in {\mathbb R}$ on~$\partial_t \Omega$ as well as a set of grain-averaged strains $E = \{E_g \} \in ({\mathbb R}^{N_\mathrm{s}})^{N_\mathrm{g}}$, where $N_\mathrm{s} = d(d+1)/2$ is the number of strain components.  We seek to find a displacement field~$u$ that satisfies the equations of elasticity and applied boundary conditions,
\begin{subequations} \label{eq:elas}
\begin{align}
\nabla \cdot {\mathbb C} \nabla u = 0 \quad &\text{in } \Omega,\label{eq:elas_domain} \\
({\mathbb C} \nabla u)n = t \quad &\text{on } \partial_t \Omega, \label{eq:elas_load}\\
({\mathbb C} \nabla u)n = 0 \quad &\text{on } \partial \Omega \setminus \partial_t \Omega, \label{eq:elas_zero_load} \\
\int_{\Omega} u \ \mathrm{d}V = \int_{\Omega} x \times u \ \mathrm{d}V = 0, \quad & \label{eq:elas_constraint}
\end{align}
\end{subequations}
for some equilibrated traction distribution $t \in \mathcal{T}_F$, defined as
\begin{equation}\label{eq:admissible_tractions}
  \mathcal T_F = \bigg\{t \in L^2(\partial_t \Omega) \mathrel{\bigg|} \int_{\partial^\mathrm{t}_t\Omega} e \cdot t \, \mathrm{d}S = - \int_{\partial^\mathrm{b}_t\Omega} e \cdot t \, \mathrm{d}S = F, \, \int_{\partial_t\Omega} t \, \mathrm{d}S = \int_{\partial_t\Omega} x \times t \, \mathrm{d}S = 0 \bigg\},
\end{equation}
consistent with the total applied force~$F$ in direction~$e$ as well as the observed grain-averaged strains:
\begin{equation} \label{eq:ave}
E_g = \frac{1}{2}\int_{\Omega_g} (\nabla u + \nabla u^\top) \, \mathrm{d}V, \quad g=1, \dots, N_\mathrm{g}.
\end{equation}
For future use, we also define the inner-product space
\begin{equation}\label{eq:t}
  \mathcal T = \bigg\{t \in L^2(\partial_t \Omega) \mathrel{\bigg|} \int_{\partial_t\Omega} t \, \mathrm{d}S = \int_{\partial_t\Omega} x \times t \, \mathrm{d}S = 0 \bigg\},
\end{equation}
which is the space of all tractions with zero net force and moment.

We are interested in three questions: 1) does the problem have a solution, that is, does there exist a traction distribution that induces the observed grain-averaged strains; 2) if so, is the solution unique; and 3) if the solution is non-unique, can the non-uniqueness be quantified?  We begin by answering the first two here and address the third in Section~\ref{sec:saint_venants_principle}.

\begin{theorem} \label{th:0inf}
There are either zero or an infinite number of displacement fields consistent with the applied force~$F$, grain-averaged strains~$E$, and the equations of elasticity (\ref{eq:elas}).
\end{theorem}

\begin{proof}
We begin by noting that given any $t \in \mathcal{T}$, there exists a unique displacement field~$u \in H^1(\Omega)$ that satisfies the equation of elasticity (\ref{eq:elas})~\cite{gurtinLinearTheoryElasticity1973}, and that the solution~$u$ is linear in~$t$.  So, we define the linear solution operator
\begin{equation}\label{eq:solution_operator}
\mathcal{S}\colon t \in \mathcal{T} \mapsto u \in H^1(\Omega) \text{ that satisfies (\ref{eq:elas})}.
\end{equation}
Conversely, given any $u$ that satisfies Equations~\eqref{eq:elas_domain}, \eqref{eq:elas_zero_load}, and \eqref{eq:elas_constraint}, there exists $t \in \mathcal{T}$ defined by Equation~\eqref{eq:elas_load}.  We also define the linear grain-averaging operator 
\begin{equation}
\mathcal{G}\colon u \in H^1 (\Omega) \mapsto E \in ({\mathbb R}^{N_\mathrm{s}})^{N_\mathrm{g}} \text{ that satisfies (\ref{eq:ave})}.
\end{equation}
Thus, given any traction $t\in \mathcal{T}$ we can find the grain-averaged strains as 
\begin{equation} \label{eq:comp}
E = (\mathcal{G} \circ \mathcal{S}) t = \mathcal{L} t.
\end{equation}

Putting these together, the problem of finding the strain field from grain-averaged values may be re-formulated as the following inverse problem: given $F$~and~$E$, find $t\in \mathcal{T}_F$ that satisfies Equation~\eqref{eq:comp}.  Since Equation~\eqref{eq:comp} is linear, its solution(s) depends on the properties of the operator $\mathcal{L}$. As $\mathcal{L}$~is a map from an infinite-dimensional space to a finite-dimensional space, $\ker (\mathcal{L})$ has infinite dimension. 

We define a {\it kernel field} as a traction and corresponding strain field such that there is zero net load, moment, and grain-averaged strains, i.e., $t \in \ker (\mathcal{L}) \cap \mathcal{T}_0$, where $\mathcal T_0$ is defined as $\mathcal T_F$ with $F=0$.  Note that $\mathcal T_0$~can be written as $\mathcal T_0 = \ker \mathcal C$ where $\mathcal C\colon L^2 \to \mathbb{R}^{N_\mathrm{s} + 1}$ is a linear operator that computes the net forces and moments of a traction distribution. We can then define $\tilde{\mathcal{L}} \colon L^2 \to \mathbb{R}^{N_\mathrm{s} N_\mathrm{g}} \oplus \mathbb{R}^{N_\mathrm{s} + 1}$ as $\tilde{\mathcal{L}} t := (\mathcal Lt, \mathcal Ct)$ such that $\ker \tilde{\mathcal{L}} = \ker (\mathcal{L}) \cap \mathcal{T}_0$. Following the same dimensionality argument, $\tilde{\mathcal{L}}$ also has an infinite-dimensional kernel; hence, $\ker (\mathcal{L}) \cap \mathcal{T}_0$ has infinite dimension. This argument shows that there are an infinite number of kernel fields; therefore, if one solution to Equation~\eqref{eq:comp} exists, there are an infinite number of solutions.

It is possible that the problem has no solution, as demonstrated by the following counterexample. Consider a homogeneous isotropic material ($N_\mathrm{g}=1$) with Young's modulus~$Y > 0$, Poisson's ratio~$\nu>0$, and observed hydrostatic strain state: $E = \varepsilon I$.  There is no elasticity solution that satisfies a zero traction condition on the lateral surfaces.
\end{proof}

\begin{remark}
We can reformulate the problem to one of finding an approximate reconstruction, 
\begin{equation}
\min_{t \in {\mathcal T}_F} \ \| \mathcal{L} t - E \|_2^2.
\end{equation}
This also has an infinite number of solutions due to kernel fields.  So, one can try to find the best-approximate reconstruction by penalizing the $L^2$-norm of the tractions:
\begin{equation}\label{eq:regularized_continuous_lsq}
\min_{t \in {\mathcal T}_F} \ \| \mathcal{L} t - E \|_2^2 + \lambda^2 \| t \|_2^2,
\end{equation}
for some penalty~$\lambda > 0$.  The solution to Equation~\eqref{eq:regularized_continuous_lsq} is unique.  However, the unique solution depends on~$\lambda$, and one can create other criteria to resolve the non-uniqueness.  
\end{remark}

\section{Examples} \label{sec:examples}

\subsection{Numerical implementation}\label{sec:discrete_matrices}

Following the problem statement defined in Section~\ref{sec:problem_formulation}, we now develop a numerical approach to explicitly solve the problem with the finite-element method. We discretize the domain using a finite-element discretization with $N_\mathrm{n}$~total nodes, $N_\mathrm{f}$~nodes where nodal forces are prescribed, and $N_\mathrm{e}$~elements.  In this discretization, the linear operator~$\mathcal{L}$ mapping loads to grain-averaged strains is an $N_\mathrm{s} N_\mathrm{g}  \times d N_\mathrm{f}$ matrix~$L$, defined as
\begin{equation}\label{eq:discrete_L}
  L := G B K^{-1} T,
\end{equation}
where $G$~computes the grain-averaged strains from the element strains, $B$~is a globally assembled version of the local symmetric gradient matrix defined from nodal shape functions, $K$~is the standard globally assembled finite-element stiffness matrix, and $T$~is a mapping matrix taking boundary loads~$f$ to the global nodal force vector.  These matrices are explicitly defined in Appendix~\ref{appendix:discrete_matrix_definitions}.

We would like to limit ourselves to equilibrated loads or the discretization of ${\mathcal T}$.  So we define an $N_\mathrm{s} + 1 \times d N_\mathrm{f}$ constraint matrix~$C$ such that $Cf=c_F$ prescribes that $f$~is an equilibrated surface nodal force vector with net force~$F$.  In other words $\{ f : Cf=c_F \}$ is the discretization of~$\mathcal{T}_F$. The constraint matrix~$C$ is also defined explicitly in Appendix \ref{appendix:discrete_matrix_definitions}.

The discretized problem of recovering the full-field strains becomes the problem of finding $f$ such that $E = Lf$ and $Cf=c_F$. As in Theorem~\ref{th:0inf}, if the discretization is sufficiently fine, there are either zero or a large number of solutions.  So, we study the problem of finding the best-approximate reconstruction,
\begin{equation}\label{eq:damped_projection_constrained_least_squares}
  \min_{ \{f : \ Cf=c_F\} } \|Lf - E\|^2_2 + \frac{\lambda^2}{F^2}\|f\|^2_2,
\end{equation}
which is the discretized form of Equation~\eqref{eq:regularized_continuous_lsq}. It is convenient to reformulate Equation~\eqref{eq:damped_projection_constrained_least_squares}. Choosing $P$ as the projection matrix onto $\ker C$ and setting $f_F$ to be the uniform traction field corresponding to the total prescribed load $F$, we may then rewrite Equation~\eqref{eq:damped_projection_constrained_least_squares} as an unconstrained regularized linear least-squares problem:
\begin{equation}\label{eq:bas}
  \min_{f} \|L(Pf+f_F) - E\|^2_2 + \frac{\lambda^2}{F^2}\|f\|^2_2.
\end{equation}
Equation~\eqref{eq:bas} may be solved using any desired least-squares solver to provide a unique\footnote{The unique solution is dependent on the chosen value of the regularization parameter~$\lambda$.} solution, i.e., the best-approximate full-field reconstruction.

We solve this problem using the LSMR method~\cite{fongLSMRIterativeAlgorithm2011}.  This is an iterative method that does not require an explicit matrix representation of~$L$.  Instead, the method simply requires the action of $L$~and~$L^\top$ (once each per iteration). For highly refined finite-element meshes, this allows for the use of iterative system solvers instead of explicit multiplication by the inverse stiffness matrix~$K^{-1}$; this evaluation is the most computationally expensive part of computing the action of $L$~and~$L^\top$.

\subsection{Nonuniqueness and kernel traction distribution}\label{sec:kernel_example}
\begin{figure}
  \centering
  \includegraphics[width=6.5in]{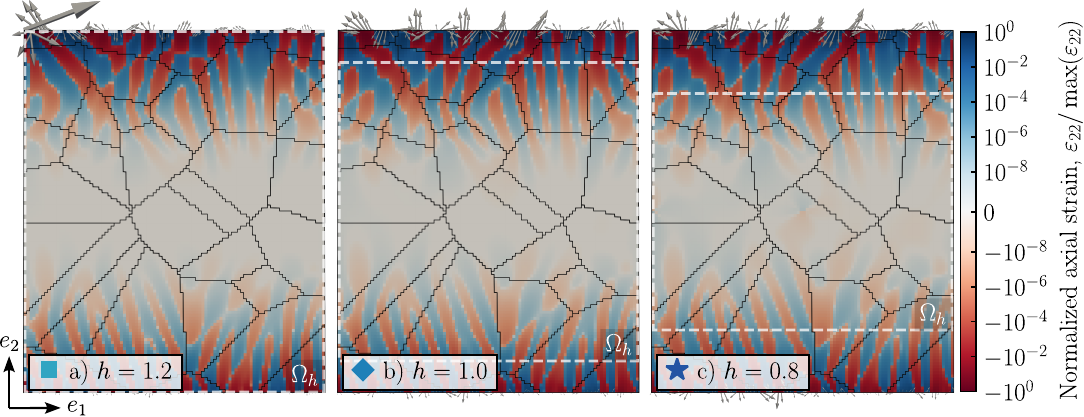}
  \caption{Three examples of kernel tractions (zero applied force and zero grain-averaged strains) and their corresponding strain fields in a unitless 2D domain with size $1 \times 1.2$. The different examples correspond to kernel tractions inducing maximal uncertainty in the subdomain~$\Omega_h$ of height~$h$ (denoted by the  dashed white boxes) for (a)~$h=1.2$ (whole domain), (b)~$h=1.0$, and (c)~$h=0.8$, see Section~\ref{sec:saint_venants_principle}. The markers correspond to those shown in Figure~\ref{fig:max_eigenvalues}.}
  \label{fig:null_traction}
\end{figure}

We begin by providing examples of kernel fields.  Recall that a kernel traction distribution~$t$ is an equilibrated traction corresponding to zero applied force and zero grain-averaged strains ($t \in {\mathcal T}_0 \cap \ker {\mathcal L}$ in the continuous setting or $t \in \ker C \cap \ker L$ in the discrete setting).  The existence of a kernel traction distribution is equivalent to non-uniqueness since we can add any multiple of a kernel traction distribution to a solution to obtain another solution.

We provide three examples of kernel loads and their corresponding strain fields in two dimensions~($d=2$).   We consider a unitless example with a rectangular domain of size $1 \times 1.2$ consisting of 33 grains with shapes obtained by Voronoi tessellation and orientations assigned with a random angle~$\theta \in [0, 2\pi)$.  We use the following example parameters: Young's modulus~$Y=1$, Poisson's ratio~$\nu=0.3$, and cubic anisotropy parameter~$\alpha=1$. We discretize the domain with a finite-element grid consisting of $100 \times 120$ square quadrilateral elements.  We explicitly assemble the matrices $L$ and $C$, and compute their joint kernel using singular value decomposition (SVD).  
Figures~\ref{fig:null_traction}(a--c) show three examples of kernel tractions with the axial normal strain~($\varepsilon_{22}$) field displayed. The grain-averaged strains of the strain fields shown in Figures~\ref{fig:null_traction}(a--c) are identically zero by construction; the same is true for the other two strain components.

In each of these examples, we see that the tractions fluctuate significantly and the strain fields decay exponentially away from the top and bottom boundaries.  We discuss this further in Section~\ref{sec:saint_venants_principle}.

\subsection{Best-approximate reconstruction}

\begin{figure}[!htbp]
  \centering
  \includegraphics[width=5.75in]{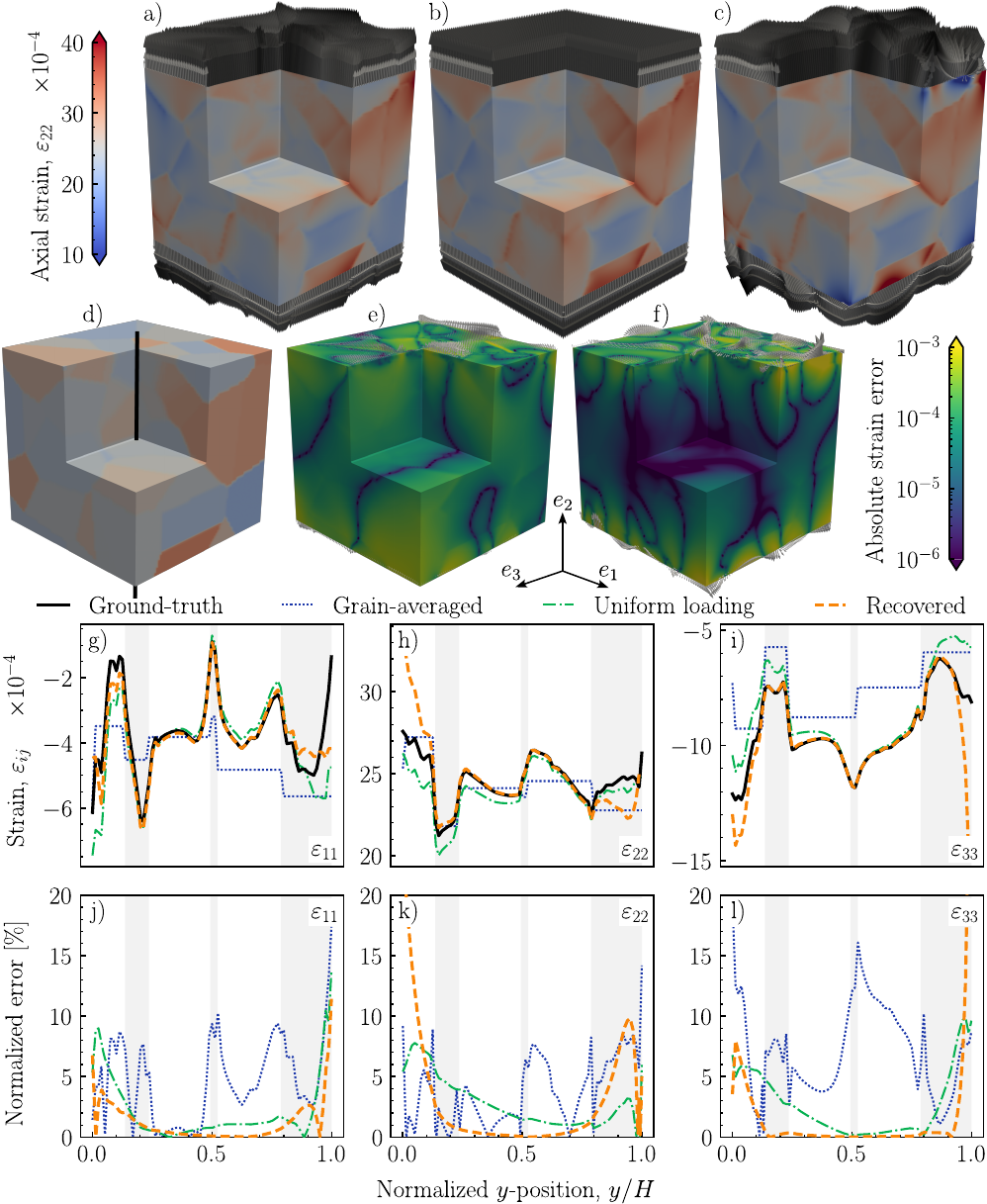}\\
  \caption{Best-approximate reconstruction using synthetic data. (a)~``Ground truth'', (b)~uniform loading initial guess, and (c)~best-approximate reconstruction loads and axial strain distributions. (d)~Grain-averaged strains used to solve the inverse problem. (e)~Uniform loading and (f)~best-approximate reconstruction error/difference between ground truth (viz., difference between (a)~and~(b); (a)~and~(c)). (g)~$\varepsilon_{11}$, (h)~$\varepsilon_{22}$, and (i)~$\varepsilon_{33}$ strain components plotted along the vertical black line shown in~(d) ($x_2 \in (0, H)$, $x_1=x_3=H/2$) for the four fields shown in~(a--d). (j--l)~Normalized error (normalized by the mean axial strain)  between recovered and ground-truth strain fields for the normal strain components along the same line. In~(g--l), changes in grains are highlighted by the alternating lightly shaded regions.}
  \label{fig:3d_strain_field_comparison}
\end{figure}

We study the best-approximate reconstruction using synthetic data. We consider a synthetic polycrystal consisting of 77 grains in the domain of interest with orientation and strain fields obtained following a procedure detailed in Appendix~\ref{appendix:ground_truth_dataset}. The resulting strain field is shown in Figure~\ref{fig:3d_strain_field_comparison}(a).  We then average the strain over the grains to obtain our ``data'', shown in Figure~\ref{fig:3d_strain_field_comparison}(d).  

We solve Equation~\eqref{eq:bas} with $\lambda=0.02$ and $f_F$ as the initial guess (uniform loading) to obtain the best-approximate loading. The loads and the corresponding strain fields are shown in Figures~\ref{fig:3d_strain_field_comparison}(b) and~\ref{fig:3d_strain_field_comparison}(c).  The difference between these and the ``ground-truth'' used to generate the data are shown in Figures~\ref{fig:3d_strain_field_comparison}(e) and~\ref{fig:3d_strain_field_comparison}(f).  In the whole domain, we observe similar levels of relative error\footnote{We define relative error between the ground-truth strain field~$\varepsilon_\text{ground}(x)$ and some other strain field~$\varepsilon(x)$ as $\|\varepsilon_\text{ground}(x) - \varepsilon(x)\|_2 / \|\varepsilon_\text{ground}(x)\|_2 \cdot 100 \%$.} between the best-approximate and uniform loading cases (7.8\%~vs.~7.1\%, respectively); however, in the center half of the domain (i.e., $x_2 \in (H/4, 3H/4)$), the error in the best-approximate loading case is much lower than the uniform case (0.7\%~vs.~3.6\%, respectively).

Figures~\ref{fig:3d_strain_field_comparison}(g--i) compare the strain components (both grain-averaged and full-field) along the vertical center line ($\{H/2, (0,H), H/2\}$), while Figures~\ref{fig:3d_strain_field_comparison}(j--l) show the errors.  Again, we find good agreement between the ground truth and the best-approximate reconstruction except near the top and bottom boundaries; this result is further expanded in Section~\ref{sec:saint_venants_principle}. We also note from Figures~\ref{fig:3d_strain_field_comparison}(j--l) that the best-approximate reconstruction has lower error in the center of the domain than the uniform loading case.

\subsection{Application to experimental data}
\begin{figure}
  \centering
  \includegraphics[width=6.5in]{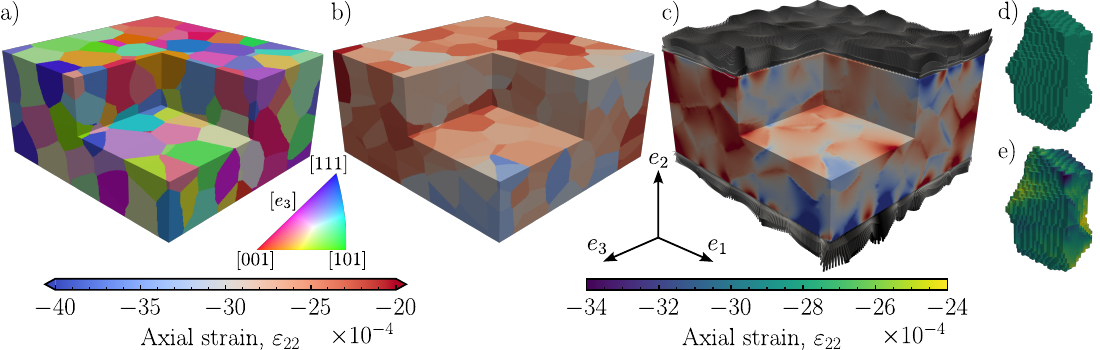}
  \caption{Best-approximate reconstruction for experimental data in AlON.  (a)~Grain structure (inverse pole figure coloring) and (b)~grain-averaged axial strains as measured from ff-HEDM analysis. (c)~Best-approximate reconstruction. (d)~Experimentally measured grain-averaged axial strains and (e)~associated best-approximate intragranular strain field in a grain fully contained within the volume.}
  \label{fig:experimental_full_field_strains}
\end{figure}

  \begin{table}
    \centering
    \caption{Cubic elastic constants for aluminum oxynitride (AlON)~\cite{satapathySingleCrystalElasticProperties2016}.}
    \label{table:elastic_properties}
    \begin{tabular}{l l l}
      \hline
      $c_{11}$ [GPa] & $c_{12}$ [GPa] & $c_{44}$ [GPa] \\
      \hline
      334.8          & 164.4          & 178.6          \\
      \hline
    \end{tabular}
  \end{table}

We apply the method of finding the best-approximate reconstruction to experimental observations collected from an aluminum oxynitride (AlON) parallelepiped under compression.  AlON is a linearly elastic transparent ceramic widely used in blast-proof windows~\cite{mccauleyAlONBriefHistory2009}.  It has a face-centered cubic lattice (space group Fd$\overline 3$m), a strain-free lattice constant of 7.947\,\AA~as determined by powder X-ray diffraction analysis (see Appendix~\ref{appendix:further_experimental_details}), and the elastic constants reported in Table~\ref{table:elastic_properties}.  A polycrystalline sample was subjected to uniaxial compression with a total compressive load of $F = -1700\,\text{N}$, and an approximately $1370 \times 700 \times 1460\, \mu$m region was interrogated using ff-HEDM at the 1-ID beamline at the Argonne National Laboratory Advanced Photon Source.  The grain radii, grain centroids, grain orientations, and grain-averaged elastic strains were obtained from the diffraction patterns using MIDAS~\cite{sharmaFastMethodologyDetermine2012a,sharmaFastMethodologyDetermine2012}.  The grain structure was approximated using a grain-radius-weighted Voronoi tessellation based on the grain centroids\footnote{A near-field measurement that provides accurate grain structure was not made.}.  The grain structure is shown in Figure~\ref{fig:experimental_full_field_strains}(a), and the grain-averaged strains are shown in Figure~\ref{fig:experimental_full_field_strains}(b).  

We find the best-approximate full-field strain reconstruction by solving Equation~\eqref{eq:bas} with a mesh size of $137 \times 70 \times 146$ brick elements and $\lambda=0.5$. The reconstruction is shown in Figure~\ref{fig:experimental_full_field_strains}(c), with further details given in Appendix~\ref{appendix:further_experimental_details}. The best-approximate full-field strain reconstruction elucidates the significantly heterogeneous strain field and stress concentrations within each grain, cf.\ Figures~\ref{fig:experimental_full_field_strains}(d) and~\ref{fig:experimental_full_field_strains}(e), and fluctuations in the loads and strain fields in this experimental case are similar to those seen in the numerical case (cf., Figures~\ref{fig:experimental_full_field_strains}(c) and~\ref{fig:3d_strain_field_comparison}(c)). It is also worth noting that the full-field strain is not known in this case; hence, the best-approximate reconstruction obtained here cannot be directly validated.  However, comparisons of this method to other experimental or reconstruction techniques are possible. This is left for future work.

\section{Quantifying nonuniqueness uncertainty}\label{sec:saint_venants_principle}
\begin{figure}
  \centering
  \includegraphics[width=6.5in]{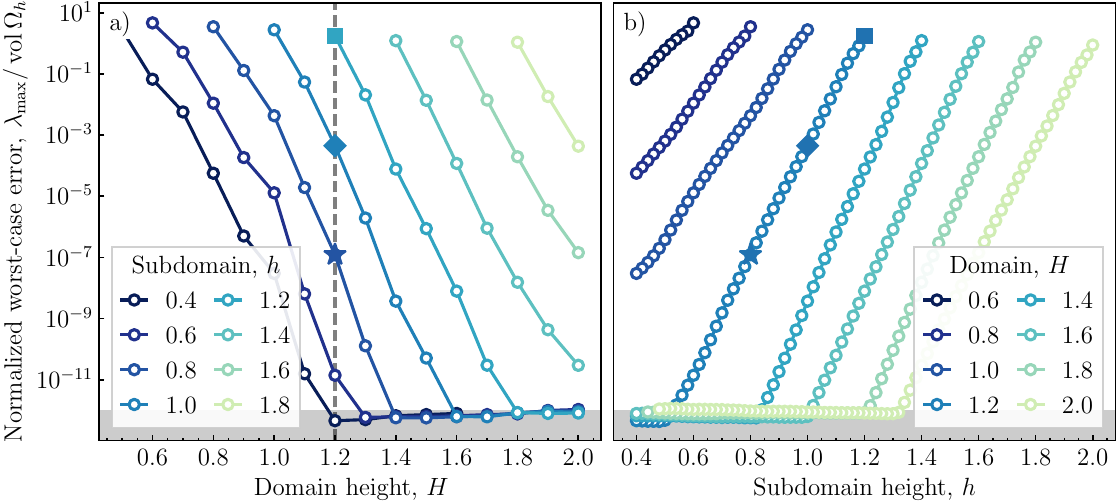}
  \caption{Normalized worst-case error (maximum eigenvalue normalized by volume) in a subdomain of height~$h$ versus the domain height~$H$. Worst-case error as a function of (a)~domain height for fixed subdomain heights, and (b)~subdomain height for fixed domain heights. The shaded regions approximately indicate where numerical breakdown occurs, and the non-circular markers correspond to the fields in Figure \ref{fig:null_traction}(a--c).}
  \label{fig:max_eigenvalues}
\end{figure}

We have shown in Section~\ref{sec:problem_formulation} that the reconstruction of the full-field strains from the grain-averaged values is non-unique due to the presence of an infinite number of kernel fields (traction distribution with net zero force and corresponding strain fields with zero grain-wise average).  So we examine the nature of the kernel fields and the uncertainty they may lead to in the full-field reconstruction.

The examples we studied earlier involved cylindrical domains with traction-free lateral surfaces and tractions applied to the ends; viz., they correspond to Saint-Venant problems.  Further, we recall from Figure~\ref{fig:null_traction} that kernel tractions are highly fluctuating and that kernel strain fields decay rapidly away from the boundaries on which tractions are applied.  We also recall from Figure~\ref{fig:3d_strain_field_comparison} that the strain error was the highest near the surfaces on which the tractions are applied.  This is reminiscent of Saint-Venant's principle \cite{toupinSaintVenantsPrinciple1965,knowlesSaintVenantPrincipleTwodimensional1966,toupinStVenantPrinciple1966}, which states that in cylindrical domains loaded at the ends, the strain becomes uniform away from the ends. Rigorous statements by Toupin~\cite{toupinSaintVenantsPrinciple1965} prove that deviations from uniformity decay exponentially away from the boundaries in homogeneous isotropic elastic cylinders in two and three dimensions.  We are unaware of any rigorous results in heterogeneous polycrystals.

Motivated by the examples in Section~\ref{sec:examples} and Saint-Venant's principle, we look at the magnitude of the kernel strain fields in the center section of a cylindrical domain of height~$H$.  Specifically, we take our cylindrical domain to be $\Omega = \omega \times (-H/2,H/2)$ where $\omega$~is the cross section.  We set $\Omega_h = \omega \times (-h/2,h/2)$ as a truncated cylindrical subdomain of height~$h < H$. We then seek to find the unit kernel field with maximal $L^2$-norm in this subdomain, i.e., the maximum uncertainty due to nonuniqueness:
\begin{equation}\label{eq:continuous_max_error}
\begin{aligned}
\lambda_\text{max} &= \max_\text{unit kernel fields} \ \int_{\Omega_h} \left| \frac{1}{2}(\nabla u + \nabla u^\top) \right|^2 \, \mathrm{d}x, \\
 &=\max_{t \in \mathcal{T}_0 \cap \ker{\mathcal L}, \|t\|_{2} =1}  \ \int_{\Omega_h} \left| \nabla \mathcal{S} t + (\nabla \mathcal{S} t)^\top\right|^2 \, \mathrm{d}x,
\end{aligned}
\end{equation} 
where $\mathcal S$~maps the surface tractions to displacements (see Equation~\eqref{eq:solution_operator}).

In the discretized setting, Equation~\eqref{eq:continuous_max_error} becomes
\begin{equation}\label{eq:discrete_max_error}
\begin{aligned}
\lambda_\text{max} &= \max_{f \in \ker C \cap \ker L, \|f\| = 1} \ f^\top M f, \\
 &=\max_{f} \ \frac{ f^\top \tilde P M \tilde P f}{\|f\|^2} ,
\end{aligned}
\end{equation} 
where the Gramian matrix $M := T^\top K^{-\top} B^{\top} B K^{-1} T \in \mathbb R^{d N_\mathrm{f} \times d N_\mathrm{f}}$ such that $f^\top M f$ is analogous to the integral in Equation~\ref{eq:continuous_max_error}, and $\tilde P$ is any projection onto $\ker C \cap \ker L$.  We conclude that the relative uncertainty due to non-uniqueness is bounded by the largest eigenvalue of~$\tilde P M \tilde P$.

We compute $\lambda_\text{max}$ for various values of $H$~and~$h$ in two dimensions for the parameters and given in Section~\ref{sec:kernel_example} and the polycrystal shown in Figure~\ref{fig:null_traction}, the results of which are shown in Figure~\ref{fig:max_eigenvalues}.  We observe that for any fixed~$h$, the maximum error decays exponentially with increasing~$H$ (Figure~\ref{fig:max_eigenvalues}(a)); and for any fixed~$H$, the maximum error decays exponentially with decreasing~$h$ (Figure~\ref{fig:max_eigenvalues}(b)).  In other words, the worst-case error drops exponentially as we move away from the ends.  This is consistent with the results of Saint-Venant's principle~\cite{toupinSaintVenantsPrinciple1965}.   More significantly, it implies that all solutions of the strain reconstruction will more closely agree in the central section than the top and bottom boundaries, and that the strain reconstruction from grain-wise averages can be meaningful in the central section as long as the data collection and reconstruction are performed on a sufficiently long region.

\section{Conclusion}\label{sec:conclusions}

We have addressed the problem of recovering admissible full-field strain fields from grain-averaged strains and crystal orientations of an elastic polycrystal measured using high-energy X-ray diffraction microscopy~(HEDM).  We have shown that this problem has either zero or an infinite number of solutions.  Specifically, there are an infinite number of kernel fields, boundary tractions with zero net force and zero grain-averaged strains.  We address non-existence by defining a best-approximate reconstruction.  We describe a numerical method and illustrate these results with a series of examples, including an example using experimental data collected on an aluminum oxynitride~(AlON) coupon.  We then show, inspired by Saint-Venant's principle, that the kernel fields decay exponentially away from the ends of the loaded ends, and consequently, the strain reconstruction from grain-wise averages can be meaningful in the central segment for sufficiently tall interrogation volumes.  

In this work, we address grain-averaged strains, but the ideas and results are the same for voxel-averaged strains, relevant to emerging techniques like point-focused HEDM~(pf-HEDM).

This work, like others, considered the problem of recovering full-field strains as a post-processing step.  However, one can envision a more direct approach.  As the crystallographic orientations and elastic strains are obtained from diffraction measurements by solving an inverse problem (i.e., reconstruction), one can extend existing reconstruction methods by adding constraints to the physics of deformation in the form of compatibility, equilibrium, and constitutive relations.  This extension would add to the computational cost, but would provide a more accurate assessment of the actual strain fields.

\section*{Acknowledgements}
We thank Peter Voorhees for numerous discussions during the course of this work.  We gratefully acknowledge the financial support of the US Office of Naval Research (No.~N00014-21-1-2784). We also gratefully acknowledge the support of the US National Science Foundation through Graduate Research Fellowships to CKC and SFG (No.~2139433).

\appendix
\section{Additional details of numerical method}\label{appendix:discrete_matrix_definitions}

\paragraph{Polycrystal}
We generate our polycrystals using Voronoi tessellation based on Euclidian distance starting from $N_\mathrm{g}$~seed points sampled uniformly.  Each resulting grain is assigned a random orientation drawn uniformly from the unit 4-sphere, and the stiffness tensor in the global coordinate system for each grain is assigned by rotating the crystal frame stiffness tensor~$\mathbb C^\mathrm{cryst}$, which is populated with the cubic elastic constants of AlON given in Table~\ref{table:elastic_properties}, where $\mathbb C^\mathrm{cryst}$~is given in reduced Voigt notation as
\begin{equation}
    \mathbb C^{\mathrm{cryst}} = \begin{bmatrix}
        c_{11} & c_{12} & c_{12} & 0 & 0 & 0 \\
        c_{12} & c_{11} & c_{12} & 0 & 0 & 0 \\
        c_{12} & c_{12} & c_{11} & 0 & 0 & 0 \\
        0 & 0 & 0 & c_{44} & 0 & 0 \\
        0 & 0 & 0 & 0 & c_{44} & 0 \\
        0 & 0 & 0 & 0 & 0 & c_{44} \\
    \end{bmatrix}.
\end{equation}
In the case that the anisotropic axis form is used, the stiffness tensor in the crystal frame is given by:
\begin{equation}
    \mathbb C^\mathrm{cryst}_{ijkl} = \frac{Y \nu}{(1 + \nu)(1 - 2\nu)}\delta_{ij}\delta_{kl} + \frac{Y}{2(1+\nu)}(\delta_{ik}\delta_{jl} + \delta_{il}\delta_{jk}) + \alpha \left(\sum_{n=1}^d e^n_i e^n_j e^n_k e^n_l\right),
\end{equation}
where $\delta_{ij}$~is the Kronecker delta, $Y$~is Young's modulus, $\nu$~is Poisson's ratio, $\alpha$~is the cubic anisotropy parameter, and $e^n_i$~is the $n$th basis vector.

\paragraph{Discretization}
We discretize our domain using P1 Lagrange elements for all simulations (affine hexahedral and quadrilateral elements for 3D and 2D, respectively). All finite-element-related operations are conducted using the open-source deal.II finite-element library~\cite{arndtDealIILibrary2023}. A total of $N_\mathrm{s}$~degrees of freedom are prescribed with homogeneous Dirichlet boundary conditions to remove all translational and rotational invariance; this ensures uniqueness of the finite-element system while not overconstraining the problem.
 
\paragraph{Matrices}
The matrices in Section~\ref{sec:discrete_matrices} are explicitly defined as follows.  First, define the extension matrix mapping prescribed nodal forces to the full finite-element forcing vector, $T \in \mathbb R^{d N_\mathrm{n} \times d N_\mathrm{f}}$, as
  \begin{equation}
    T_{ij} = \begin{cases}
      1 & S(j) = i,   \\
      0 & \text{else},
    \end{cases}
  \end{equation}
  where $S \colon \mathbb{N} \to \mathbb{N}$ is an injection from the boundary DOF index to the nodal DOF index. The finite-element stiffness matrix $K \in \mathbb R^{d N_\mathrm{n} \times d N_\mathrm{n}}$ then relates the global nodal force vector~$f_\mathrm{RHS} = Tf$ and displacement vector~$u$ as $Ku = f_\mathrm{RHS}$. For a properly constrained problem, $K$~is full rank and its inverse $K^{-1} \in \mathbb R^{d N_\mathrm{n} \times d N_\mathrm{n}}$ exists. We then define the global symmetric gradient matrix $B \in \mathbb R^{N_\mathrm{s} N_\mathrm{e} \times d N_\mathrm{n}}$, which maps discrete nodal displacements to element strains, as the globally assembled version of the element strain--displacement matrix $B_\mathrm{elem} \in \mathbb R^{N_\mathrm{s} \times d 2^{d}}$ defined using gradients of the local element shape functions $\phi$; e.g., for 2D affine quadrilateral elements, $B_\mathrm{elem}$ is defined as
  \begin{equation}
    B_\mathrm{elem} = \begin{bmatrix}
      \frac{\partial \phi_1}{\partial x_1}   & 0                                      & \frac{\partial \phi_2}{\partial x_1}   & 0                                      & \frac{\partial \phi_3}{\partial x_1} & 0                                      & \frac{\partial \phi_4}{\partial x_1}   & 0                                      \\
      0                                      & \frac{\partial \phi_1}{x_2}            & 0                                      & \frac{\partial \phi_2}{x_2}            & 0                                    & \frac{\partial \phi_3}{x_2}            & 0                                      & \frac{\partial \phi_4}{x_2}            \\
      \frac{1}{2}\frac{\partial \phi_1}{x_2} & \frac{1}{2}\frac{\partial \phi_1}{x_1} & \frac{1}{2}\frac{\partial \phi_2}{x_2} & \frac{1}{2}\frac{\partial \phi_2}{x_1} & \frac{\partial \phi_3}{x_2}          & \frac{1}{2}\frac{\partial \phi_3}{x_1} & \frac{1}{2}\frac{\partial \phi_4}{x_2} & \frac{1}{2}\frac{\partial \phi_4}{x_1}
    \end{bmatrix},
  \end{equation}
  noting that tensorial shear strains are computed here instead of the typical engineering shear strains. The strain in linear quadrilateral elements is not constant, so the weighted-average over all quadrature points is taken. Finally, the grain-averaging matrix $G \in \mathbb R^{N_\mathrm{s} N_\mathrm{g} \times N_\mathrm{s} N_\mathrm{e}}$ is defined as
  \begin{equation}
    G_{ij} = \begin{cases}
      \operatorname{vol} (E_j) / \operatorname{vol} (\Omega_i) & E^\mathrm{c}_j \in \Omega_i, \\
      0                          & \text{else},
    \end{cases}
  \end{equation}
  where $E_j$~and~$E^\mathrm{c}_j$ are the $j$th element and its centroid, respectively, and $\Omega_i$~is the domain of grain~$i$. In practice, $T$,~$B$, and~$G$ are assembled as sparse matrices.

  The matrix~$C \in \mathbb{R}^{N_\mathrm{s} + 1 \times d N_\mathrm{f}}$ is a globally assembled (over the surfaces with applied loading) form of $C_\mathrm{node} \in \mathbb R^{N_\mathrm{s} + 1 \times d}$ which computes the net force and moment due to the applied force~$f_\mathrm{node} \in \mathbb{R}^{d}$ at the node, which, assuming that the total force~$F$ is applied in the $e_2$-direction, is defined in 3D as
  \begin{equation}\label{eq:constraint_matrix}
    C_\mathrm{node} = \begin{bmatrix}
      1    & 0        & 0    \\
      0    & -I(x)    & 0    \\
      0    & 1 - I(x) & 0    \\
      0    & 0        & 1    \\
      0    & -x_3     & x_2  \\
      x_3  & 0        & -x_1 \\
      -x_2 & x_1      & 0
    \end{bmatrix}, \quad c_F = \begin{bmatrix}
      0 \\ F \\ F \\ 0 \\ 0 \\ 0 \\ 0
    \end{bmatrix},
  \end{equation}
  where $I(x) = 1$ if $x \in \partial^\mathrm{t}_t \Omega$ and $I(x)=0$ otherwise. When $C$~is applied to~$f$, the first four rows compute the net forces while the bottom three rows compute the net moments. The uniform top and bottom normal load vector~$f_F$ with total load~$F$ then trivially satisfies $Cf_F = c_F$. To optimize memory usage, the projection matrix~$P$ onto~$\ker C$ is not explicitly assembled; rather, its action is computed as $P = I - QQ^\top$ where the columns of~$Q$ span $\operatorname{coim}C = (\ker C)^\perp$, obtained via the reduced QR factorization: $C^\top = QR$. In Equation~\eqref{eq:discrete_max_error} $L$ is explicitly formed, so $\tilde P$~is computed from the null space basis vectors obtained from the SVD of the stacked matrix $\begin{bmatrix} C & L \end{bmatrix}^\top$.

\paragraph{Solvers}
A conjugate-gradient solver with an algebraic multigrid preconditioner implemented in Trilinos~\cite{thetrilinosprojectteamTrilinosProjectWebsite} is used to solve the elasticity problem.   The LSMR solver~\cite{fongLSMRIterativeAlgorithm2011} is used to solve the least-squares problem (Equation~\eqref{eq:bas}).  This algorithm only requires the action of~$L$ (and its transpose), which, as defined in Equation~\eqref{eq:discrete_L}, requires the application of~$K^{-1}$.  In three dimensions, $L$~is not explicitly formed and the action of $K^{-1} = K^{-\top}$ is computed using the preconditioned conjugate-gradient solver. Further, only the conjugate-gradient solve is parallelized while the action of~$G$,~$B$,~$T$,~$P$, and the iterative least-squares solver updates are all computed by a single primary thread. In two dimensions, $L$~is formed discretely and $K^{-1}$ is directly computed. In \emph{all} computations, units are rescaled to be non-dimensionalized for improved numerical behavior, i.e., $\hat f = f / F$, $\hat{\mathbb{C}} = \mathbb{C} / c_{11}$, and $\hat \Omega = \Omega / \ell$ where $\ell$~is the characteristic length of~$\Omega$. In strain computations and visualizations, the units are scaled back.

\section{Numerical ground-truth dataset}\label{appendix:ground_truth_dataset}

A square cuboidal domain $\Omega_{3H} = (0,H) \times (-H, 2H) \times (0,H)$ is generated with $N_\mathrm{g} = 150$ grains; see Figure \ref{fig:ground_truth_dataset}(a).  The linear elasticity problem with uniform applied loading on the top and bottom surfaces is solved using the finite-element method with a mesh size of $80 \times 240 \times 80$ cubic brick elements; see Figure~\ref{fig:ground_truth_dataset}(b). The microstructure and grain-averaged strains are then extracted from a cubic subdomain $\Omega_H = (0,H) \times (0, H) \times (0,H) \subset \Omega_{3H}$ ($80 \times 80 \times 80$ cubic brick elements) in the center of the larger volume; the internal force on the boundary $\partial \Omega_H \cap \Omega_{3H}$ is also extracted; see Figure \ref{fig:ground_truth_dataset}(c). The extracted microstructure region is used as a synthetic reference dataset against which the best-approximate reconstruction is tested. Physical units are then assigned to the problem: $H = 1\,\text{mm}$, $F = 850\,\text{N}$, and the cubic elastic constants of AlON are used (Table \ref{table:elastic_properties}). 
\begin{figure}[!htbp]
    \centering
    \includegraphics[width=5in]{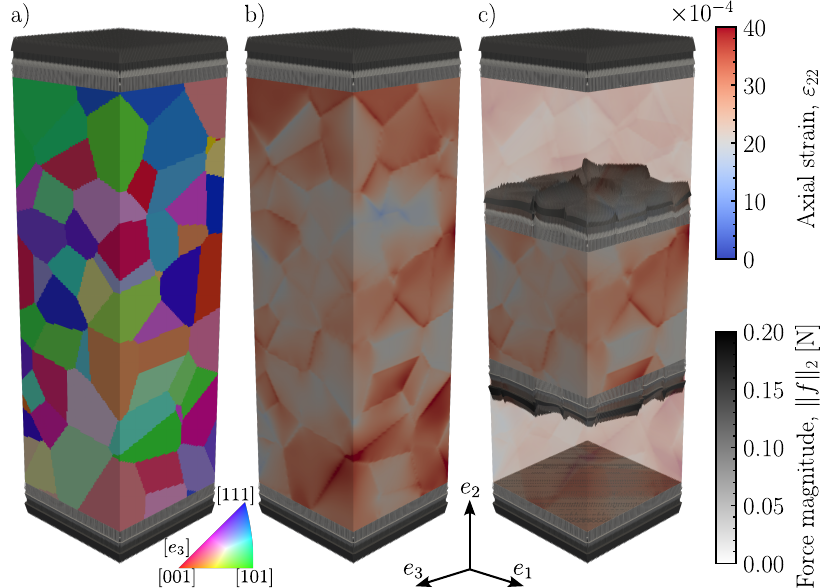}
    \caption{a)~Microstructure generated by Voronoi tessellation (inverse pole figure coloring). b)~Axial strain field of the uniformly loaded microstructure as determined from a linear-elastic finite-element computation. c)~Extracted subdomain axial strain field with the associated non-uniform internal boundary loading.}
    \label{fig:ground_truth_dataset}
  \end{figure}

\section{Further details on experimental comparison}\label{appendix:further_experimental_details}

  \begin{figure}[!htbp]
    \centering
    \includegraphics[width=6.5in]{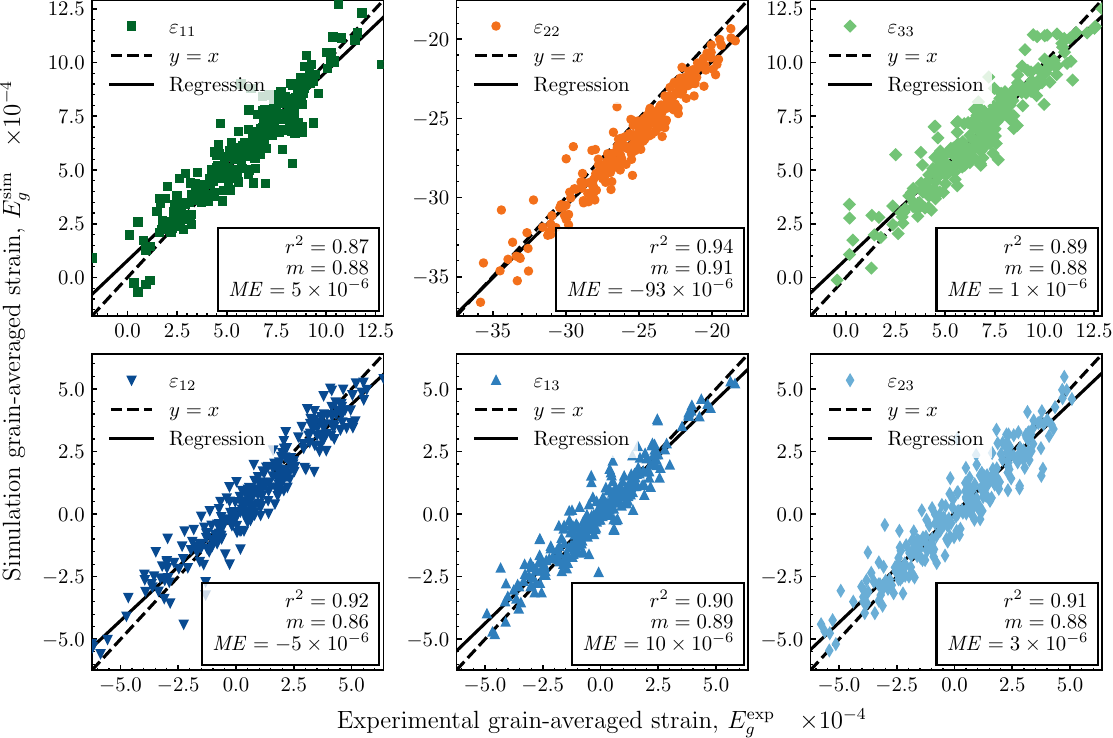}
    \caption{Plot showing the recovered versus corrected experimental grain-averaged strains for all grains within the experimental volume. For each plot, linear regression is performed, and the regression line is plotted along with the coefficient of determination~$r^2$, slope~$m$, and mean error~$\mathit{ME}$. If $Lf = E^\mathrm{exp}$ could be solved exactly, the points would all fall along the line~$y=x$.}
    \label{fig:experimental_regression_plot}
  \end{figure}

 We noticed a hydrostatic offset between the experimentally measured and simulated grain-averaged strains under uniform load; this is a consequence of an incorrect value for the strain-free lattice parameter of AlON ~\cite{reischigThreedimensionalReconstructionIntragranular2020}.  To roughly correct for the offset, we computed the grain-averaged strains from a finite-element simulation with uniform loading applied to the boundaries of the microstructure.  The mean hydrostatic offset is determined to be $\varepsilon_\mathrm{h} = -9.32 \times 10^{-4}$ by comparing these strains to the measured grain-averaged strains. This offset is removed from all experimentally measured normal strains, and these ``corrected'' strains are those with which the inverse problem was solved.  Figure~\ref{fig:experimental_regression_plot} shows the recovered versus corrected experimental grain-averaged strains along with the linear regression between them.  If the hydrostatic offset were not removed, the data points for the normal strain components would all fall below the lines~$y=x$ in Figure~\ref{fig:experimental_regression_plot}.

Note the mean slope of $m_\mathrm{avg} = 0.88$ in the regression lines seen in Figure \ref{fig:experimental_regression_plot}; this indicates that the experimentally measured strains have a larger scaling than the recovered strains. The tessellated microstructure likely has smoother grain boundaries than in the actual experiment and small grains may not have been found by the far-field indexing procedure; hence, the experimental microstructure shown in Figure~\ref{fig:experimental_full_field_strains}(a) likely has less heterogeneity than reality.

\end{document}